\documentclass[12pt,draftcls,onecolumn]{IEEEtran}

\usepackage{cite}      

\usepackage[dvips]{color}
\usepackage{amssymb}
\usepackage{amsmath}
\usepackage{graphicx}
\usepackage{algorithmic}
\newtheorem{theorem}{Theorem}
\newtheorem{lemma}{Lemma}
\newtheorem{proposition}{Proposition}
\newtheorem{corollary}{Corollary}
\newtheorem{definition}{Definition}

\newtheorem{remark}{Remark}
\newcommand{\R}{\mathbb{R}}
\usepackage{setspace}
\newcommand{\longversion}[1]{}

\renewcommand{\longversion}[1]{}

\begin{document}

\title{Robust Regression and Lasso}
\author{Huan~Xu\IEEEmembership{},~Constantine~Caramanis,~\IEEEmembership{Member},~and~
Shie~Mannor,~\IEEEmembership{Member} \thanks{A preliminary version
of this paper was presented at Twenty-Second Annual Conference on
Neural Information Processing Systems.}
 \thanks{H. Xu and S. Mannor are with the Department of Electrical and
Computer Engineering, McGill University, Montr{\'e}al, H3A2A7,
Canada email: (xuhuan@cim.mcgill.ca;
shie.mannor@mcgill.ca)}
\thanks{C. Caramanis is with the Department of Electrical and Computer
Engineering, The University of Texas at Austin, Austin, TX 78712 USA
email: (cmcaram@ece.utexas.edu).} \markboth{~Oct.~2008}{Shell
\MakeLowercase{\textit{et al.}}: Bare Demo of IEEEtran.cls for
Journals}} \maketitle
\begin{abstract}
Lasso, or $\ell^1$ regularized least squares, has been explored
extensively for its remarkable sparsity properties. It is shown in
this paper  that the solution to Lasso, in addition to its sparsity,
has robustness properties: it is the solution to a robust
optimization problem. This has two important consequences. First,
robustness provides a connection of the regularizer to a physical
property, namely, protection from noise. This allows a principled
selection of the regularizer, and in particular, generalizations of
Lasso that also yield convex optimization problems are obtained by
considering different uncertainty sets.

Secondly, robustness can itself be used as an avenue to exploring
different properties of the solution. In particular, it is shown
that robustness of the solution  explains why the solution is
sparse. The analysis as well as the specific results obtained differ
from standard sparsity results, providing different geometric
intuition. Furthermore, it is shown that the robust optimization
formulation is related to kernel density estimation, and based on
this approach,  a proof that Lasso is consistent is given using
robustness directly. Finally, a  theorem saying that sparsity and
algorithmic stability contradict each other, and hence Lasso is not
stable, is presented.
\end{abstract}

%

\section{Introduction}
In this paper we consider linear regression problems with
least-square error. The problem is to find a vector $\mathbf{x}$ so that
the $\ell_2$ norm of the residual $\mathbf{b}-A\mathbf{x}$ is minimized,
for a given matrix $A\in \mathbb{R}^{n\times m}$ and vector
$\mathbf{b}\in \mathbb{R}^n$. From a learning/regression
perspective, each row of $A$ can be regarded as a training sample,
and the corresponding element of $b$ as the target value of this
observed sample. Each column of $A$ corresponds to a feature, and
the objective is to find a set of weights so that the weighted sum
of the feature values approximates the target value.

It is well known that minimizing the least squared error can lead to
sensitive solutions \cite{Elden85,Golub89,Higham92,Fierro94}.
Many regularization methods have been proposed to decrease this sensitivity.
Among them, Tikhonov regularization \cite{Tikhonov77} and
Lasso~\cite{Tibshirani96,Efron04} are two widely known and cited algorithms.
These methods minimize a weighted sum of the residual norm and a
certain regularization term, $\|\mathbf{x}\|_2$ for Tikhonov
regularization and $\|\mathbf{x}\|_1$ for Lasso. In addition to providing
regularity, Lasso is also known for the tendency to select sparse solutions.
Recently this has attracted much attention for its ability to reconstruct sparse
solutions when sampling occurs far below the Nyquist rate, and also for its
ability to recover the sparsity pattern exactly with probability one, asymptotically
as the number of observations increases (there is an extensive literature on this
subject, and we refer the reader to
\cite{Chen98,Feuer03,Candes06,Tropp04,Wainwright2006} and references therein).

The first result of this paper is that the solution to Lasso has
robustness properties: it is the solution to a robust optimization
problem. In itself, this interpretation of Lasso as the solution to
a robust least squares problem is a development in line with the
results of \cite{Ghaoui97}. There, the authors propose an
alternative approach of reducing sensitivity of linear regression by
considering a robust version of the regression problem, i.e.,
minimizing the worst-case residual for the observations under some
unknown but bounded disturbance. Most of the research in this area
considers either the case where the disturbance is row-wise
uncoupled \cite{Shivaswamy06}, or the case where the Frobenius norm
of the disturbance matrix is bounded \cite{Ghaoui97}.

None of these robust optimization approaches produces a solution
that has sparsity properties (in particular, the solution to Lasso
does not solve any of these previously formulated robust
optimization problems). In contrast, we investigate the robust
regression problem where the uncertainty set is defined by
feature-wise constraints. Such a noise model is of interest when
values of features are obtained with some noisy pre-processing
steps, and the magnitudes of such noises are known or bounded.
Another situation of interest is where features are meaningfully
coupled.  We define {\it coupled} and {\it uncoupled} disturbances
and uncertainty sets precisely in Section \ref{ssec:formulation}
below. Intuitively, a disturbance is feature-wise coupled if the
variation or disturbance across features satisfy joint constraints,
and uncoupled otherwise.

Considering the solution to Lasso as the solution  of a robust least
squares problem has two important consequences. First, robustness
provides a connection of the regularizer to a physical property,
namely, protection from noise. This allows more principled selection
of the regularizer, and in particular, considering different
uncertainty sets, we construct generalizations of Lasso that also
yield convex optimization problems.

Secondly, and perhaps most significantly, robustness  is a strong
property that can itself be used as an avenue to investigating
different properties of the solution. We show that robustness of the
solution can explain why the solution is sparse. The analysis as
well as the specific results we obtain differ from standard sparsity
results, providing different geometric intuition, and extending
beyond the least-squares setting. Sparsity results obtained for
Lasso ultimately depend on the fact that introducing additional
features incurs larger $\ell^1$-penalty than the least squares error
reduction. In contrast, we exploit the fact that a robust solution
is, by definition, the optimal solution under a worst-case
perturbation. Our results show that, essentially, a coefficient of
the solution is nonzero if the corresponding feature is relevant
under all allowable perturbations. In addition to sparsity, we also
use robustness directly to prove consistency of Lasso.
%

We briefly list the main contributions as well as the organization
of this paper.
\begin{itemize}
\item In Section~\ref{sec.robustregressionformulation}, we formulate the robust regression problem with feature-wise independent disturbances, and show that this
formulation is equivalent to a least-square problem with a weighted
$\ell_1$ norm regularization term. Hence, we provide an
interpretation
of Lasso from a robustness perspective.
\item We generalize the robust
 regression formulation to loss functions of arbitrary norm in Section~\ref{sec.generalization}.
 We also consider uncertainty sets
that require disturbances of different features to satisfy joint
conditions. This can be used to mitigate the conservativeness of the
robust solution  and to obtain solutions with additional
properties. 

\item In
Section~\ref{sec.sparse}, we present new sparsity results for the
robust regression problem with feature-wise independent
disturbances. This provides a new robustness-based explanation to
the sparsity of Lasso. Our approach gives new analysis and also
geometric intuition, and furthermore allows one to obtain sparsity
results for more general loss functions, beyond the squared loss.

\item Next, we relate Lasso to kernel density estimation in Section~\ref{sec.density}. This allows us to re-prove consistency
in a statistical learning setup, using the new robustness tools and
formulation we introduce. Along with our results on sparsity, this
illustrates the power of robustness in explaining and also exploring
different properties of the solution.

\item Finally, we prove in Section~\ref{sec.stability}
a ``no-free-lunch'' theorem, stating that an algorithm that
encourages sparsity cannot be stable.
\end{itemize}
{\bf Notation}. We use capital letters to represent matrices, and
boldface letters to represent column vectors. Row vectors are
represented as the transpose of column vectors. For a vector
$\mathbf{z}$, $z_i$ denotes its $i^{th}$ element. Throughout the
paper, $\mathbf{a}_i$ and $\mathbf{r}_j^\top$ are used to denote the
$i^{th}$ column and the $j^{th}$ row of the observation matrix $A$,
respectively. We use $a_{ij}$ to denote the $ij$ element of $A$,
hence it is the $j^{th}$ element of $\mathbf{r}_i$, and $i^{th}$
element of $\mathbf{a}_j$. For a convex function $f(\cdot)$,
$\partial f(\mathbf{z})$ represents any of its sub-gradients
evaluated at $\mathbf{z}$. A vector with length $n$ and each element
equals $1$ is denoted as $\mathbf{1}_n$.

\section{Robust Regression with Feature-wise Disturbance}\label{sec.robustregressionformulation}
In this section, we show that our robust regression formulation recovers Lasso as a special case. We also derive
probabilistic bounds that guide in the construction of the uncertainty set.

The regression formulation we consider differs from the standard
Lasso formulation, as we minimize the norm of the error, rather than
the squared norm. It is known that these two coincide up to a change
of the regularization coefficient. Yet as we discuss above, our
results lead to more flexible and potentially powerful robust
formulations,  and give new insight into known results.

\subsection{Formulation}
\label{ssec:formulation}
Robust linear regression considers the case where the
observed matrix is corrupted by some potentially malicious disturbance.
The objective is to find the optimal solution in the worst case sense.
This is usually formulated as the following min-max problem,
\begin{equation}\begin{split}\label{equ.robustregression}&\mbox{{\bf Robust Linear Regression:}}\\ &\min_{\mathbf{x}\in
\mathbb{R}^m}\left\{\max_{\Delta A\in \mathcal{U}}
\|\mathbf{b}-(A+\Delta
A)\mathbf{x}\|_2\right\},\end{split}
\end{equation}
where $\mathcal{U}$ is called the {\it uncertainty set}, or the set of admissible disturbances of the matrix $A$.
In this section, we consider the class of uncertainty sets that bound the norm of the disturbance to each feature,
without placing any joint requirements across feature disturbances. That is, we consider the class of uncertainty
sets:
\begin{equation}\label{equ.featuredisturbance}\mathcal{U}\triangleq\Big\{(\boldsymbol{\delta}_1,\cdots,\boldsymbol{\delta}_m)\Big|\|\boldsymbol{\delta}_i\|_2\leq
c_i,\,\,\,i=1,\cdots, m\Big\},
\end{equation}
for given $c_i\geq 0$. We call these uncertainty sets {\it
feature-wise uncoupled}, in contrast to {\it coupled} uncertainty
sets that require disturbances of different features to satisfy some
joint constraints (we discuss these extensively below, and their
significance). While the inner maximization problem of
(\ref{equ.robustregression}) is nonconvex, we show in the next
theorem that uncoupled norm-bounded uncertainty sets lead to an
easily solvable optimization problem.
\begin{theorem}\label{them.l2ci}The robust regression
problem~(\ref{equ.robustregression}) with uncertainty set of the form
(\ref{equ.featuredisturbance}) is equivalent to the following $\ell^1$
regularized regression problem:
\begin{equation}\label{equ.reg_reg}
\min_{\mathbf{x}\in \mathbb{R}^m}\Big\{
\|\mathbf{b}-A\mathbf{x}\|_2+\sum_{i=1}^m c_i|x_i|\Big\}.
\end{equation}
\end{theorem}
\begin{proof}
Fix $\mathbf{x}^*$. We prove that $\max_{\Delta A\in \mathcal{U}}
\|\mathbf{b}-(A+\Delta A)\mathbf{x}^*\|_2 =
\|\mathbf{b}-A\mathbf{x}^*\|_2+\sum_{i=1}^m c_i|x^*_i|$.

The left hand side can be written as
\begin{equation}\label{equ.proofsmall}
\begin{split}
&\max_{\Delta A\in \mathcal{U}} \|\mathbf{b}-(A+\Delta
A)\mathbf{x}^*\|_2
\\=& \max_{(\boldsymbol{\delta}_1,\cdots, \boldsymbol{\delta}_m)|\|\boldsymbol{\delta}_i\|_2\leq
c_i} \Big\|\mathbf{b}-\big(A+(\boldsymbol{\delta}_1,\cdots,
\boldsymbol{\delta}_m)\big)\mathbf{x}^*\Big\|_2
\\=& \max_{(\boldsymbol{\delta}_1,\cdots, \boldsymbol{\delta}_m)| \|\boldsymbol{\delta}_i\|_2\leq
c_i}  \|\mathbf{b}-A\mathbf{x}^*-\sum_{i=1}^m
x_i^*\boldsymbol{\delta}_i \|_2
\\\leq & \max_{(\boldsymbol{\delta}_1,\cdots, \boldsymbol{\delta}_m)| \|\boldsymbol{\delta}_i\|_2\leq
c_i} \Big\|\mathbf{b}-A\mathbf{x}^*\Big\|_2+\sum_{i=1}^m \|x_i^*
\boldsymbol{\delta}_i\|_2
\\\leq &  \|\mathbf{b}-A\mathbf{x}^* \|_2+\sum_{i=1}^m |x_i^*|
c_i.
\end{split}
\end{equation}
Now, let
\[\mathbf{u}\triangleq \left\{\begin{array}{ll}
\frac{\mathbf{b}-A\mathbf{x}^*}{\|\mathbf{b}-A\mathbf{x}^*\|_2} &
\mbox{if}\,\, A\mathbf{x}^*\not=\mathbf{b},\\ \mbox{any vector with
unit $\ell^2$ norm} &\mbox{otherwise};\end{array}\right.
\]
and let
\[\boldsymbol{\delta}^*_i \triangleq -c_i \mathrm{sgn}(x_i^*) \mathbf{u}.\] Observe that
$\|\boldsymbol{\delta}^*_i\|_2\leq c_i$, hence $\Delta
A^*\triangleq(\boldsymbol{\delta}^*_1,\cdots,
\boldsymbol{\delta}^*_m)\in \mathcal{U}$. Notice that
\begin{equation}\label{equ.prooflarge}
\begin{split}
&\max_{\Delta A\in \mathcal{U}} \|\mathbf{b}-(A+\Delta
A)\mathbf{x}^*\|_2\\
 \geq&\|\mathbf{b}-(A+\Delta A^*)\mathbf{x}^*\|_2\\
=&\Big\|\mathbf{b}-\big(A+(\boldsymbol{\delta}_1^*,\cdots,\boldsymbol{\delta}_m^*)\big)\mathbf{x}^*\Big\|_2\\
=&\Big\|(\mathbf{b}-A\mathbf{x}^*) -\sum_{i=1}^m \big(-x_i^* c_i
\mathrm{sgn}(x_i^*)\mathbf{u}\big)\Big\|_2\\
=&\Big \|(\mathbf{b}-A\mathbf{x}^*)+(\sum_{i=1}^m c_i
|x_i^*|)\mathbf{u}\Big\|_2\\
=& \|\mathbf{b}-A\mathbf{x}^*\|_2 +\sum_{i=1}^m c_i |x_i^*|.
\end{split}
\end{equation}
The last equation holds from the definition of $\mathbf{u}$.

Combining Inequalities~(\ref{equ.proofsmall}) and
(\ref{equ.prooflarge}), establishes the equality
$\max_{\Delta A\in \mathcal{U}} \|\mathbf{b}-(A+\Delta
A)\mathbf{x}^*\|_2 = \|\mathbf{b}-A\mathbf{x}^*\|_2+\sum_{i=1}^m
c_i|x^*_i|$ for any $\mathbf{x}^*$. Minimizing over $\mathbf{x}$ on both sides
proves the theorem.
\end{proof}
Taking $c_i=c$  and normalizing $\mathbf{a}_i$ for
all $i$, Problem (\ref{equ.reg_reg}) recovers the well-known Lasso
\cite{Tibshirani96,Efron04}.

\subsection{Uncertainty Set Construction}
The selection of an uncertainty set $\mathcal{U}$ in Robust Optimization is of fundamental importance.
One way this can be done is as an approximation of so-called {\it chance constraints}, where
a deterministic constraint is replaced by the requirement that a constraint is satisfied with at
least some probability. These can be formulated when we know the distribution exactly,
or when we have only partial information of the uncertainty, such as, e.g., first and second moments.
This chance-constraint formulation is particularly important when the distribution has large support,
rendering the naive robust optimization formulation overly pessimistic.

 For confidence level $\eta$, the chance constraint formulation becomes:
\begin{equation*}\begin{split}
\mbox{minimize:}\quad &t\\\mbox{Subject to:} \quad&
\mathrm{Pr}(\|\mathbf{b}-(A+{\Delta A})\mathbf{x}\|_2 \leq t)\geq
1-\eta.
\end{split}\end{equation*}
Here, $\mathbf{x}$ and $t$ are the
decision variables.

Constructing the uncertainty set for feature $i$ can be done quickly via
line search and bisection, as long as we can evaluate $\mathrm{Pr}(\|\mathbf{a}_i\|_2\geq c)$.
If we know the distribution exactly (i.e., if we have complete probabilistic information),
this can be quickly done via sampling. Another setting of
interest is when we have access only to some moments of the distribution of the uncertainty,
e.g., the mean and variance. In this setting, the uncertainty sets are constructed via a
bisection procedure which evaluates the worst-case probability over all distributions with
given mean and variance. We do this using a tight bound on the probability of an
event, given the first two moments.
%

%



In the scalar case, the Markov Inequality provides such a bound. The
next theorem is a generalization of the Markov inequality to
$\mathbb{R}^n$, which bounds the probability where the disturbance
on a given feature is more than $c_i$, if only the first and second
moment of the random variable are known. We postpone the proof to
the appendix, and refer the reader to \cite{BertsimasPopescu04} for
similar results using semi-definite optimization.
\begin{theorem}\label{thm.probabound}Consider a random vector $\mathbf{v}\in \mathbb{R}^n$,
such that $\mathbb{E}(\mathbf{v})=\mathbf{a}$, and
$\mathbb{E}(\mathbf{v} \mathbf{v}^\top)=\Sigma$, $\Sigma \succeq 0$.
Then we have
\begin{equation}\label{equ.chancebound}\Pr\{\|\mathbf{v}\|_2\geq c_i\} \leq
\left\{\begin{array}{rl} \min_{P, \mathbf{q}, r,
\lambda}\,\,&\mathrm{Trace}(\Sigma P)+2\mathbf{q}^\top\mathbf{a}+
r\\\mbox{subject to:}\,\,&\left(\begin{array}{ll}P & \mathbf{q} \\ \mathbf{q}^\top & r\end{array}\right)\succeq 0\\
 &\left(\begin{array}{ll} I(m) &
\mathbf{0} \\ \mathbf{0}^\top &-c^2_i\end{array}\right)\preceq
\lambda \left(\begin{array}{ll} P & \mathbf{q}\\ \mathbf{q}^\top &
r-1\end{array} \right)\\ &\lambda \geq
0.\end{array}\right.\end{equation}
\end{theorem}

The optimization problem (\ref{equ.chancebound}) is a semi-definite
programming, which is known be solved in polynomial time.
Furthermore, if we replace $\mathbb{E}(\mathbf{v}\mathbf{v}^\top)=
\Sigma$ by an inequality $\mathbb{E}(\mathbf{v}\mathbf{v}^\top)\leq
\Sigma$, the uniform bound still holds. Thus, even if our estimation
to the variance is not precise, we are still able to bound the
probability of having ``large'' disturbance.

\section{General Uncertainty Sets}\label{sec.generalization}
One reason the robust optimization formulation is powerful, is that having provided
the connection to Lasso, it then allows the opportunity to generalize to efficient
``Lasso-like'' regularization algorithms.

In this section, we make several generalizations of the robust
formulation~(\ref{equ.robustregression}) and derive counterparts of
Theorem~\ref{them.l2ci}. We generalize the robust formulation in two
ways: (a) to the case of arbitrary norm; and (b) to the case of
coupled uncertainty sets.

We first consider the case of
an arbitrary norm $\|\cdot\|_a$ of $\mathbb{R}^n$ as a cost function
rather than the squared loss.
The proof of the next theorem is identical to that of
Theorem~\ref{them.l2ci}, with only the $\ell^2$ norm changed to
$\|\cdot\|_a$.

\begin{theorem}\label{thm.generalnorm}
The robust regression problem
\[\min_{\mathbf{x}\in \mathbb{R}^m} \left\{\max_{\Delta A \in
\mathcal{U}_a}\|\mathbf{b}-(A+\Delta
A)\mathbf{x}\|_a\right\};\,\,\,\mathcal{U}_a\triangleq
\Big\{(\boldsymbol{\delta}_1,\cdots,\boldsymbol{\delta}_m)\Big|\|\boldsymbol{\delta}_i\|_a\leq
c_i,\,\,\,i=1,\cdots, m\Big\};\] is equivalent to the following
regularized regression problem
\[\min_{\mathbf{x}\in \mathbb{R}^m}\Big\{
\|\mathbf{b}-A\mathbf{x}\|_a+\sum_{i=1}^m c_i|x_i|\Big\}.\]
\end{theorem}

We next remove the assumption that the disturbances are feature-wise
uncoupled. Allowing coupled uncertainty sets is useful when we have
some additional information about potential noise in the problem,
and we want to limit the conservativeness of the worst-case
formulation.
Consider the following uncertainty set:
\[\mathcal{U}'\triangleq \left\{(\boldsymbol{\delta}_1,\cdots, \boldsymbol{\delta}_m)\big| f_j(\|\boldsymbol{\delta}_1\|_a,\cdots, \|\boldsymbol{\delta}_m\|_a)\leq 0;\,\, j=1,\cdots, k \right\},\]
where $f_j(\cdot)$ are convex functions. Notice that, both $k$ and
$f_j$ can be arbitrary, hence  this is a very general formulation,
and provides us with significant flexibility in designing
uncertainty sets and equivalently new regression algorithms (see for
example Corollary~\ref{cor.generalnorm} and~\ref{cor.polytopeu}).
The following theorem converts this formulation to tractable
optimization problems. The proof is postponed to the appendix.
\begin{theorem}\label{thm.generaluncertainty}
Assume that the set \[\mathcal{Z}\triangleq\{\mathbf{z}\in
\mathbb{R}^m| f_j(\mathbf{z})\leq 0,\,\,j=1,\cdots, k;\,\,
\mathbf{z}\geq \mathbf{0}\}\] has non-empty relative interior. Then the
robust regression problem
\[\min_{\mathbf{x}\in \mathbb{R}^m} \left\{\max_{\Delta A \in
\mathcal{U}'}\|\mathbf{b}-(A+\Delta A)\mathbf{x}\|_a\right\}\] is
equivalent to the following regularized regression problem
\begin{equation}\label{equ.generaluncertainty}\begin{split}&\min_{\boldsymbol{\lambda}\in
\mathbb{R}^k_+,\boldsymbol{\kappa}\in \mathbb{R}^m_+, \mathbf{x}\in
\mathbb{R}^m}\Big\{
\|\mathbf{b}-A\mathbf{x}\|_a+v(\boldsymbol{\lambda},\boldsymbol{\kappa},
\mathbf{x})\Big\};\\&\mbox{where:}\,\,
v(\boldsymbol{\lambda},\boldsymbol{\kappa},\mathbf{x})\triangleq
\max_{\mathbf{c}\in
\mathbf{R}^m}\Big[(\boldsymbol{\kappa}+|\mathbf{x}|)^\top
\mathbf{c}-\sum_{j=1}^k\lambda_j
f_j(\mathbf{c})\Big]\end{split}\end{equation}
\end{theorem}

{\bf Remark:} Problem~(\ref{equ.generaluncertainty}) is efficiently
solvable. Denote
$z^{\mathbf{c}}(\boldsymbol{\lambda},\boldsymbol{\kappa},\mathbf{x})\triangleq\Big[(\boldsymbol{\kappa}+|\mathbf{x}|)^\top
\mathbf{c}-\sum_{j=1}^k\lambda_j f_j(\mathbf{c})\Big]$. This is a
convex function of
$(\boldsymbol{\lambda},\boldsymbol{\kappa},\mathbf{x})$, and the
sub-gradient of $z^{\mathbf{c}}(\cdot)$ can be computed easily for
any $\mathbf{c}$. The function
$v(\boldsymbol{\lambda},\boldsymbol{\kappa},\mathbf{x})$ is the
maximum of a set of convex functions, $z^{\mathbf{c}}(\cdot)$ , hence is convex,
and satisfies
\[\partial
v(\boldsymbol{\lambda}^*,\boldsymbol{\kappa}^*,\mathbf{x}^*)=\partial
z^{\mathbf{c}_0}(\boldsymbol{\lambda}^*,\boldsymbol{\kappa}^*,\mathbf{x}^*),
\] where $\mathbf{c}_0$ maximizes $\Big[(\boldsymbol{\kappa}^*+|\mathbf{x}|^*)^\top
\mathbf{c}-\sum_{j=1}^k\lambda^*_j f_j(\mathbf{c})\Big]$. We
can efficiently evaluate $\mathbf{c}_0$ due to convexity of
$f_j(\cdot)$, and hence we can efficiently evaluate the
sub-gradient of $v(\cdot)$.

The next two
corollaries are a direct application of
Theorem~\ref{thm.generaluncertainty}.
\begin{corollary}\label{cor.generalnorm}
Suppose $\mathcal{U}'=\left\{(\boldsymbol{\delta}_1,\cdots,
\boldsymbol{\delta}_m)\Big|
\big\|\|\boldsymbol{\delta}_1\|_a,\cdots,
\|\boldsymbol{\delta}_m\|_a\big\|_s\leq l; \right\}$ for a symmetric
norm $\|\cdot\|_s$, then the resulting regularized regression
problem is
\[\min_{\mathbf{x}\in \mathbb{R}^m}\Big\{
\|\mathbf{b}-A\mathbf{x}\|_a+l\|
\mathbf{x}\|_s^*\Big\};\quad\mbox{where}\,\, \|\cdot\|_s^*
\,\mbox{is the dual norm of}\,\,\|\cdot\|_s.\]
\end{corollary}
This corollary interprets {\em arbitrary} norm-based regularizers from a robust regression
perspective. For example, it is straightforward to show that if we
take both $\|\cdot\|_\alpha$ and $\|\cdot\|_s$ as the Euclidean
norm, then $\mathcal{U}'$ is the set of matrices with their
Frobenious norms bounded, and Corollary~\ref{cor.generalnorm}
reduces to the robust formulation introduced by \cite{Ghaoui97}.
\begin{corollary}\label{cor.polytopeu} Suppose $\mathcal{U}'=\left\{(\boldsymbol{\delta}_1,\cdots,
\boldsymbol{\delta}_m)\Big|\exists \mathbf{c}\geq \mathbf{0}:
T\mathbf{c}\leq \mathbf{s};\,\,  \|\boldsymbol{\delta}_j\|_a\leq
c_j; \right\}$, then the resulting regularized regression problem is
\begin{equation*}
\begin{split}
\mbox{Minimize:}\quad & \|\mathbf{b}-A\mathbf{x}\|_a
+\mathbf{s}^\top \boldsymbol{\lambda}\\
\mbox{Subject to:}\quad & \mathbf{x}\leq T^\top
\boldsymbol{\lambda}\\
&-\mathbf{x}\leq T^\top \boldsymbol{\lambda}\\
&\boldsymbol{\lambda}\geq \mathbf{0}.
\end{split}
\end{equation*}
\end{corollary}
Unlike previous results, this corollary considers general polytope
uncertainty sets. Advantages of such sets include the linearity of
the final formulation. Moreover, the modeling power is considerable,
as many interesting disturbances can be modeled in this way.

We briefly mention some further examples meant to illustrate the
power and flexibility of the robust formulation. We refer the
interested reader to \cite{LassoGerad} for full details.

As the results above indicate, the robust formulation can model a
broad class of uncertainties, and yield computationally tractable
(i.e., convex) problems. In particular, one can use the polytope
uncertainty discussed above, to show (see \cite{LassoGerad}) that by
employing an uncertainty set first used in \cite{Bertsimas04}, we
can model cardinality constrained noise, where some (unknown) subset
of at most $k$ features can be corrupted.

Another avenue one may take using robustness, and which is also
possible to solve easily, is the case where the uncertainty set
allows independent perturbation of the columns and the rows of the
matrix $A$. The resulting formulation resembles the elastic-net
formulation \cite{ZouHastie2003}, where there is a combination of
$\ell^2$ and $\ell^1$ regularization.

\section{Sparsity}\label{sec.sparse}
In this section, we investigate the sparsity properties of robust
regression~(\ref{equ.robustregression}), and equivalently Lasso.
Lasso's ability to recover sparse solutions has been extensively
studied and discussed~(cf \cite{Chen98,Feuer03,Candes06,Tropp04}).
There are generally two approaches. The first approach investigates
the problem from a statistical perspective. That is, it assumes that
the observations are generated by a (sparse) linear combination of
the features, and investigates the asymptotic or probabilistic
conditions required for Lasso to correctly recover the generative
model.
The second approach treats the problem from an optimization
perspective, and studies under what conditions a pair $(A,\,
\mathbf{b})$ defines a problem with sparse solutions (e.g.,
\cite{Tropp06}).

We follow the second approach and do not assume a generative model.
Instead, we consider the conditions that lead to a feature receiving
zero weight. Our first result paves the way for the remainder of
this section. We show in Theorem \ref{thm.nearlysame} that,
essentially, a feature receives no weight (namely, $x_i^*=0$) if
there exists an allowable perturbation of that feature which makes
it irrelevant. This result holds for general norm loss functions,
but in the $\ell^2$ case, we obtain further geometric results. For
instance, using Theorem~\ref{thm.nearlysame}, we show, among other
results, that ``nearly'' orthogonal features get zero weight
(Theorem~\ref{thm.incoherent}). Using similar tools, we provide
additional results in \cite{LassoGerad}. There, we show, among other
results, that the sparsity pattern of any optimal solution must
satisfy certain angular separation conditions between the residual
and the relevant features, and that ``nearly'' linearly dependent
features get zero weight.

Substantial research regarding sparsity properties of Lasso can be found in the literature
(cf \cite{Chen98,Feuer03,Candes06,Tropp04,Girosi98,Coifman92,Mallat93,Donoho06}
and many others). In particular, similar results as in point (a), that rely on an {\em incoherence} property, have been
established in, e.g., \cite{Tropp06}, and are used as standard tools in investigating sparsity of Lasso from the statistical
perspective. However, a proof exploiting robustness and properties of the uncertainty is novel. Indeed, such a
proof shows a fundamental connection between robustness and
sparsity, and implies that robustifying w.r.t. a feature-wise independent
uncertainty set might be a plausible way to achieve sparsity for
other problems.

To state the main theorem of this section, from which the other
results derive, we introduce some notation to facilitate the
discussion. Given a feature-wise uncoupled uncertainty set, ${\cal
U}$, an index subset $I \subseteq \{1,\dots,n\}$, and any $\Delta A
\in {\cal U}$, let $\Delta A^I$ denote the element of ${\cal U}$
that equals $\Delta A$ on each feature indexed by $i \in I$, and is
zero elsewhere. Then, we can write any element $\Delta A \in {\cal
U}$ as $\Delta A^I + \Delta A^{I^c}$ (where $I^c = \{1,\dots,n\}
\setminus I$). Then we have the following theorem. We note that the
result holds for any norm loss function, but we state and prove it
for the $\ell^2$ norm, since the proof for other norms is identical.

\begin{theorem}\label{thm.nearlysame}
The robust regression problem
\[\min_{\mathbf{x}\in
\mathbb{R}^m}\left\{\max_{\Delta A\in \mathcal{U}}
\|\mathbf{b}-(A+\Delta A)\mathbf{x}\|_2\right\},\] has a solution
supported on an index set $I$ if there exists some perturbation
$\Delta \tilde{A}^{I^c} \in \mathcal{U}$ of the features in $I^c$,
such that the robust regression problem
\[\min_{\mathbf{x}\in
\mathbb{R}^m}\left\{\max_{\Delta \tilde{A}^I \in \mathcal{U}^I}
\|\mathbf{b}-(A + \Delta \tilde{A}^{I^c}+\Delta
\tilde{A}^I)\mathbf{x}\|_2\right\},\] has a solution supported on
the set $I$.
\end{theorem}
Thus, a robust regression has an optimal solution supported on a set
$I$, if {\it any} perturbation of the features corresponding to the
complement of $I$ makes them irrelevant.
Theorem~\ref{thm.nearlysame} is a special case of the following
theorem with $c_j=0$ for all $j\not \in I$: \\ \\
{\bf Theorem \ref{thm.nearlysame}'}. {\it Let
$\mathbf{x}^*$ be an optimal solution of the robust regression
problem:
\[\min_{\mathbf{x}\in
\mathbb{R}^m}\left\{\max_{\Delta A\in \mathcal{U}}
\|\mathbf{b}-(A+\Delta A)\mathbf{x}\|_2\right\},\]
and let $I\subseteq \{1,\cdots, m\}$ be such that
$x_j^*=0$ $\forall \, j \not\in I$. Let \[\tilde{\mathcal{U}}\triangleq
\Big\{(\boldsymbol{\delta}_1,\cdots,\boldsymbol{\delta}_m)\Big|\|\boldsymbol{\delta}_i\|_2\leq
c_i,\,\,\,i \in I;\,\,\, \|\boldsymbol{\delta}_j\|_2\leq
c_j+l_j,\,\,j \not\in I\Big\}.\] Then, $\mathbf{x}^*$ is an optimal
solution of
\[\min_{\mathbf{x}\in
\mathbb{R}^m}\left\{\max_{\Delta A \in \tilde{\mathcal{U}}}
\|\mathbf{b}-(\tilde{A}+\Delta A)\mathbf{x}\|_2\right\},\]
for any $\tilde{A}$ that satisfies $\|\tilde{\mathbf{a}}_j - \mathbf{a}_j\|\leq l_j$
for $j \not\in I$, and $\tilde{\mathbf{a}}_i = \mathbf{a}_i$ for $i \in
I$.
}

\begin{proof}
Notice that
\begin{equation*}
\begin{split}
&\max_{\Delta A \in \tilde{\mathcal{U}}} \Big\|\mathbf{b}-(A+\Delta
A) \mathbf{x}^*\Big\|_2\\
=& \max_{\Delta A \in  \mathcal{U}} \Big\|\mathbf{b}-(A+\Delta
A) \mathbf{x}^*\Big\|_2\\
=& \max_{\Delta A \in \mathcal{U}}
\Big\|\mathbf{b}-(\tilde{A}+\Delta A) \mathbf{x}^*\Big\|_2.
\end{split}
\end{equation*}
These equalities hold because for $j \not\in I$, $x^*_j=0$, hence the
$j^{th}$ column of both $\tilde{A}$ and $\Delta A$ has no effect on the
residual.

For an arbitrary $\mathbf{x}'$, we have
\begin{equation*}
\begin{split}
&\max_{\Delta A \in \tilde{\mathcal{U}}} \Big\|\mathbf{b}-(A+\Delta
A)\mathbf{x}'\Big\|_2\\
\geq & \max_{\Delta A \in  \mathcal{U}}
\Big\|\mathbf{b}-(\tilde{A}+\Delta A)\mathbf{x}'\Big\|_2.
\end{split}
\end{equation*}
This is because,  $\|\mathbf{a}_j-\tilde{\mathbf{a}}_j\|\leq l_j$
for $j \not\in I$, and $\mathbf{a}_i=\tilde{\mathbf{a}}_i$ for $i \in
I$. Hence, we have
\[\big\{A+\Delta A\big|\Delta A \in
 \mathcal{U}\big\}\subseteq \big\{\tilde{A}+\Delta A\big| \Delta A\in
\tilde{\mathcal{U}}\big\}.\]

Finally, notice that
\[\max_{\Delta A \in  \mathcal{U}}
\Big\|\mathbf{b}-(A+\Delta A) \mathbf{x}^*\Big\|_2 \leq
\max_{\Delta A \in  \mathcal{U}} \Big\|\mathbf{b}-(A+\Delta
A)\mathbf{x}'\Big\|_2.\] Therefore we have \[\max_{\Delta A \in
\tilde{\mathcal{U}}} \Big\|\mathbf{b}-(\tilde{A}+\Delta A)
\mathbf{x}^*\Big\|_2\leq \max_{\Delta A \in \tilde{\mathcal{U}}}
\Big\|\mathbf{b}-(\tilde{A}+\Delta A)\mathbf{x}'\Big\|_2.\]Since this holds
for arbitrary $\mathbf{x}'$, we establish the theorem.
\end{proof}

We can interpret the result of this theorem by considering a
generative model\begin{footnote}{While we are not assuming
generative models to establish the results, it is still interesting
to see how these results can help in a generative model
setup.}\end{footnote} $b=\sum_{i\in I} w_i a_i+\tilde{\xi}$ where
$I\subseteq \{1\cdots,m\}$ and $\tilde{\xi}$ is a random variable,
i.e., $b$ is generated by features belonging to $I$. In this case,
for a feature $j \not\in I$, Lasso would assign zero weight as long
as there exists a perturbed value of this feature, such that the
optimal regression assigned it zero weight.

When we consider $\ell^2$ loss, we can translate the condition of a
feature being ``irrelevant'' into a geometric condition, namely,
orthogonality. We now use the result of Theorem~\ref{thm.nearlysame}
to show that robust regression has a sparse solution as long as an
incoherence-type property is satisfied. This result is more in line
with the traditional sparsity results, but we note that the
geometric reasoning is different, and ours is based on robustness.
Indeed, we show that a feature receives zero weight, if it is
``nearly'' (i.e., within an allowable perturbation) orthogonal to
the signal, and all relevant features.
\begin{theorem}\label{thm.incoherent}
Let $c_i=c$ for all $i$ and consider $\ell^2$ loss. If there exists
$I\subset \{1,\cdots, m\}$ such that for all $\mathbf{v}\in
\rm{span}\big(\{\mathbf{a}_i,i\in I\}\bigcup\{\mathbf{b}\}\big),\,
\|\mathbf{v}\|=1$, we have $\mathbf{v}^\top\mathbf{a}_j\leq c$,
$\forall j\not \in I$, then any optimal solution $\mathbf{x}^*$
satisfies $x^*_j=0$, $\forall j\not\in I$.
\end{theorem}
\begin{proof} For $j\not \in I$, let $\mathbf{a}_j^=$ denote the
projection of $\mathbf{a}_j$ onto the span of $\{\mathbf{a}_i,\,\,
i\in I\}\bigcup\{\mathbf{b}\}$, and let $\mathbf{a}_j^+\triangleq
\mathbf{a}_j-\mathbf{a}_j^=$. Thus, we have $\|\mathbf{a}_j^=\|\leq
c$. Let $\hat{A}$ be such that
\[\hat{\mathbf{a}}_i=\left\{\begin{array}{ll} \mathbf{a}_i & i\in I;\\ \mathbf{a}_i^+ & i\not\in
I.\end{array}\right.\] Now let
\[\hat{\mathcal{U}}\triangleq \{(\boldsymbol{\delta}_1,\cdots,\boldsymbol{\delta}_m)|\|\boldsymbol{\delta}_i\|_2\leq c, \,i\in I; \|\boldsymbol{\delta}_j\|_2=0, \, j\not \in I \}.\]
Consider the robust regression problem
$\min_{\hat{\mathbf{x}}}\Big\{\max_{\Delta A \in \hat{\mathcal{U}}}
\big\|\mathbf{b}-(\hat{A}+\Delta A)\hat{\mathbf{x}}\big\|_2\Big\}$,
which is equivalent to $\min_{\hat{\mathbf{x}}}\Big\{
\big\|\mathbf{b}-\hat{A}\hat{\mathbf{x}}\big\|_2+\sum_{i\in I}
c|\hat{x}_i|\big\}$. Note that the $\hat{\mathbf{a}}_j$ are orthogonal to
the span of $\{\hat{\mathbf{a}}_i,\,\, i\in
I\}\bigcup\{\mathbf{b}\}$. Hence for any given $\hat{\mathbf{x}}$,
by changing $\hat{x}_j$ to zero for all $j\not \in I$, the
minimizing objective does not increase.

Since $\|\hat{\mathbf{a}}-\hat{\mathbf{a}}_j\|=\|\mathbf{a}_j^=\|\leq c$
$\forall j\not\in I$, (and recall that
$\mathcal{U}=\{(\boldsymbol{\delta}_1,\cdots,\boldsymbol{\delta}_m)|\|\boldsymbol{\delta}_i\|_2\leq
c, \forall i\}$) applying Theorem~\ref{thm.nearlysame} concludes the proof.
\end{proof}

\section{Density Estimation and Consistency}\label{sec.density}
In this section, we investigate the robust linear regression
formulation from a statistical perspective and rederive {\it using
only robustness properties} that Lasso is asymptotically consistent.
The basic idea of the consistency proof is as follows. We show that
the robust optimization formulation can be seen to be the maximum
error w.r.t. a class of probability measures. This class includes a
kernel density estimator, and using this, we show that Lasso is
consistent.

\subsection{Robust Optimization, Worst-case Expected Utility and Kernel Density Estimator}
In this subsection, we present some notions and intermediate
results. In particular, we link a robust optimization formulation
with a worst expected utility (w.r.t. a class of probability
measures); we then briefly recall the definition of a kernel density
estimator. Such results will be used in establishing the consistency
of Lasso, as well as providing some additional insights on robust
optimization. Proofs are postponed to the appendix.

We first establish a general result on the equivalence between a
robust optimization formulation and a worst-case expected utility:
\begin{proposition}\label{prop.pointvsdist}Given a function $g:\mathbb{R}^{m+1}\rightarrow \mathbb{R}$
and Borel sets $\mathcal{Z}_1,\cdots, \mathcal{Z}_n\subseteq
\mathbb{R}^{m+1}$, let
\[\mathcal{P}_n\triangleq \{\mu\in \mathcal{P}| \forall S\subseteq \{1,\cdots, n\}: \mu(\bigcup_{i\in S} \mathcal{Z}_i)\geq |S|/n \}.\]
The following holds
\[\frac{1}{n}\sum_{i=1}^n \sup_{(\mathbf{r}_i,b_i)\in \mathcal{Z}_i} h(\mathbf{r}_i,b_i) =\sup_{\mu\in \mathcal{P}_n}\int_{\mathbb{R}^{m+1}} h(\mathbf{r},b)d\mu(\mathbf{r},b). \]
\end{proposition}
This leads to the following corollary for Lasso, which states that
for a given $\mathbf{x}$, the robust regression loss over the
training data is equal to the worst-case expected {\it
generalization error}.
\begin{corollary}\label{cor.distributionalrobust}Given $\mathbf{b}\in \mathbb{R}^n$, $A\in \mathbb{R}^{n\times m}$, the following equation holds for any $\mathbf{x}\in
\mathbb{R}^m$, \begin{equation}\label{equ.distributionrobust}
\|\mathbf{b}-A\mathbf{x}\|_2  +\sqrt{n}c_n\|\mathbf{x}\|_1+
\sqrt{n}c_n =\sup_{\mu\in\hat{\mathcal{P}}(n)}
\sqrt{n\int_{\mathbb{R}^{m+1}}
(b'-\mathbf{r}'^\top\mathbf{x})^2d\mu(\mathbf{r}',b')}.\end{equation}Here\begin{footnote}{Recall
that $a_{ij}$ is the $j^{th}$ element of
$\mathbf{r}_i$}\end{footnote},
\begin{equation*}
\begin{split}\hat{\mathcal{P}}(n)&\triangleq
\bigcup_{\|\boldsymbol{\sigma}\|_2\leq\sqrt{n}c_n;\,\forall i:
\|\boldsymbol{\delta}_i\|_2 \leq \sqrt{n}c_n}\mathcal{P}_n(A,
\Delta, \mathbf{b},
\boldsymbol{\sigma});\\
\mathcal{P}_n (A, \Delta, \mathbf{b},
\boldsymbol{\sigma})&\triangleq \{\mu\in \mathcal{P} |
\mathcal{Z}_i=[b_i-\sigma_i,b_i+\sigma_i]\times\prod_{j=1}^m
[a_{ij}-\delta_{ij}, a_{ij}+\delta_{ij}];\\ &\quad\quad\forall
S\subseteq \{1,\cdots, n\}: \mu(\bigcup_{i\in S} \mathcal{Z}_i)\geq
|S|/n \}.\end{split}\end{equation*}
\end{corollary}
\begin{remark}{\rm We briefly explain
Corollary~\ref{cor.distributionalrobust} to avoid possible
confusions. Equation~(\ref{equ.distributionrobust}) is a
non-probabilistic equality. That is, it holds without any assumption
(e.g., i.i.d. or generated by certain distributions) on $\mathbf{b}$
and $A$. And it does not involve any probabilistic operation such as
taking expectation on the left-hand-side, instead, it is an
equivalence relationship which hold for an arbitrary set of samples.
Notice that, the right-hand-side also depends on the samples since
$\hat{\mathcal{P}}(n)$ is defined through $A$ and $\mathbf{b}$.
Indeed, $\hat{\mathcal{P}}(n)$ represents the union of classes of
distributions
$\mathcal{P}_n(A,\Delta,\mathbf{b},\boldsymbol{\sigma})$ such that
the norm of each column of $\Delta$ is bounded, where
$\mathcal{P}_n(A,\Delta,\mathbf{b},\boldsymbol{\sigma})$ is the set
of distributions corresponds to (see
Proposition~\ref{prop.pointvsdist}) disturbance in hyper-rectangle
Borel sets $\mathcal{Z}_1,\cdots,\mathcal{Z}_n$ centered at $(b_i,
\mathbf{r}_i^\top)$ with lengths $(2\sigma_i, 2\delta_{i1},\cdots,
2\delta_{im})$. }
\end{remark}

We will later show that $\hat{P}_n$ consists a kernel density
estimator. Hence we recall here its definition.
 The {\em kernel density estimator} for a
density $\hat{h}$ in $\mathbb{R}^d$, originally proposed in
\cite{Rosenblatt56,Parzen62}, is defined by
\[h_n(\mathbf{x})=(nc_n^{d})^{-1}\sum_{i=1}^nK\left(\frac{\mathbf{x}-\hat{\mathbf{x}}_i}{c_n}\right),\]
where $\{c_n\}$ is a sequence of positive numbers,
$\hat{\mathbf{x}}_i$ are i.i.d. samples generated according to
$\hat{f}$, and $K$ is a Borel measurable function (kernel)
satisfying $K\geq 0$, $\int K=1$. See \cite{Devroye85,Scott92} and
the reference therein for detailed discussions.
Figure~\ref{fig.kerneldensity} illustrates a kernel density
estimator using Gaussian kernel for a randomly generated sample-set.
A celebrated property of a kernel density estimator is that it
converges in $\mathcal{L}^1$ to $\hat{h}$ when $c_n\downarrow 0$ and
$nc_n^{d}\uparrow \infty$ \cite{Devroye85}.
\begin{figure}[t]
\begin{center}
  \includegraphics[height=5cm, width=1\linewidth]{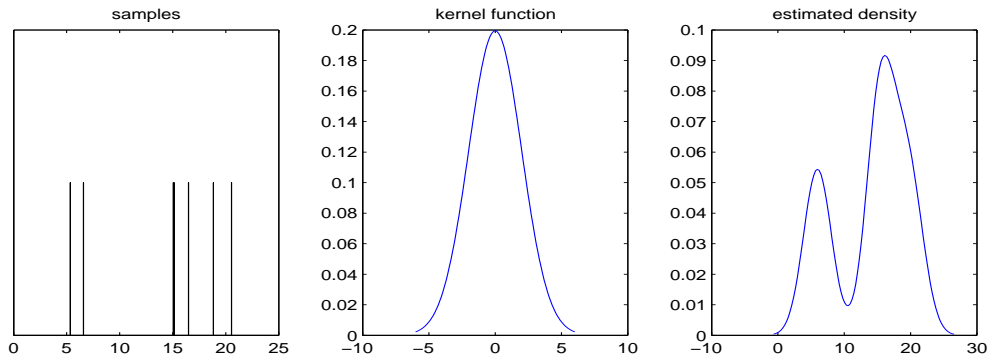}
\caption{Illustration of Kernel Density Estimation.}
\label{fig.kerneldensity}
\end{center}
\end{figure}

\subsection{Consistency of Lasso}
We restrict our discussion to the case where the magnitude of the
allowable uncertainty for all features equals $c$, (i.e., the
standard Lasso) and establish the statistical consistency of Lasso
from  a distributional robustness argument. Generalization to
the non-uniform case is straightforward.
Throughout, we use $c_n$ to represent $c$ where there are $n$ samples
(we take $c_n$ to zero).

Recall the standard generative model in statistical
learning: let $\mathbb{P}$ be a probability measure
with bounded support that generates i.i.d samples $(b_i,
\mathbf{r}_i)$, and has a density $f^*(\cdot)$. Denote the set of
the first $n$ samples by $\mathcal{S}_n$. Define
\begin{equation*}\begin{split}\mathbf{x}(c_n, \mathcal{S}_n)&\triangleq\arg\min_{\mathbf{x}}
\Big\{\sqrt{\frac{1}{n}\sum_{i=1}^n (b_i-\mathbf{r}_i^\top
\mathbf{x})^2}+c_n\|x\|_1\Big\}\\&=\arg\min_{\mathbf{x}}
\Big\{\frac{\sqrt{n}}{n}\sqrt{\sum_{i=1}^n (b_i-\mathbf{r}_i^\top
\mathbf{x})^2}+c_n\|x\|_1\Big\};\\
\mathbf{x}(\mathbb{P})&\triangleq\arg\min_{\mathbf{x}}
\Big\{\sqrt{\int_{b,\mathbf{r}} (b-\mathbf{r}^\top \mathbf{x})^2
d\mathbb{P}(b, \mathbf{r})}\Big\}.\end{split}\end{equation*} In
words, $\mathbf{x}(c_n, \mathcal{S}_n)$ is the solution to Lasso
with the tradeoff parameter set to $c_n\sqrt{n}$, and
$\mathbf{x}(\mathbb{P})$ is the ``true'' optimal solution. We have
the following consistency result. The theorem itself is a well-known
result. However, the proof technique is novel. This technique is of
interest because the standard techniques to establish consistency in
statistical learning including Vapnik-Chervonenkis (VC) dimension
(e.g., \cite{Vapnik91}) and algorithmic stability (e.g.,
\cite{Bousquet02}) often work for a limited range of algorithms,
e.g., the $k$-Nearest Neighbor is known to have infinite VC
dimension, and we show in Section \ref{sec.stability} that {\it
Lasso is not stable}. In contrast,
 a much wider range of algorithms have
robustness interpretations, allowing a unified approach to prove
their consistency.
\begin{theorem}\label{thm.consistencyforbounded}Let $\{c_n\}$
be such that $c_n\downarrow 0$ and $\lim_{n\rightarrow\infty} n
(c_n)^{m+1} =\infty$. Suppose there exists a constant $H$ such that
$\|\mathbf{x}(c_n, \mathcal{S}_n)\|_2\leq H$. Then,
\[\lim_{n\rightarrow \infty} \sqrt{\int_{b,\mathbf{r}} (b-\mathbf{r}^\top \mathbf{x}(c_n, \mathcal{S}_n))^2 d\mathbb{P}(b,
\mathbf{r})}=\sqrt{\int_{b,\mathbf{r}} (b-\mathbf{r}^\top
\mathbf{x}(\mathbb{P}))^2 d\mathbb{P}(b, \mathbf{r})},\] almost
surely.
\end{theorem}
\begin{proof}
{\bf Step 1}: We show that  the right hand side of
Equation~(\ref{equ.distributionrobust}) includes a kernel density
estimator for the true (unknown) distribution. Consider the
following kernel estimator  given samples $\mathcal{S}_n=(b_i,
\mathbf{r}_i)_{i=1}^n$  and tradeoff parameter $c_n$,
\begin{equation}\label{equ.kerneldensity}\begin{split}
&f_n(b,\mathbf{r})\triangleq (nc_n^{m+1})^{-1} \sum_{i=1}^n
K\left(\frac{b-b_i,\mathbf{r}-\mathbf{r}_i}{c_n}\right),\\
&\mbox{where: }K(\mathbf{x})\triangleq
I_{[-1,+1]^{m+1}}(\mathbf{x})/2^{m+1}.\end{split}\end{equation}

Let $\hat{\mu}_n$ denote the distribution given by the density
function $f_n(b,\mathbf{r})$. Easy to check that $\hat{\mu}_n$
belongs to $\mathcal{P}_n (A, ({c}_n\mathbf{1}_n,\cdots,
{c}_n\mathbf{1}_n), \mathbf{b}, {c}_n\mathbf{1}_n)$ and hence
belongs to $\hat{\mathcal{P}}(n)$ by definition.

{\bf Step 2}: Using the $\mathcal{L}^1$ convergence property of the
kernel density estimator, we prove the consistency of robust
regression and equivalently Lasso.

First notice that, $\|\mathbf{x}(c_n, \mathcal{S}_n)\|_2\leq H$ and
$\mathbb{P}$ has a bounded support implies that there exists a
universal constant $C$ such that
\[\max_{b,\mathbf{r}} (b-\mathbf{r}^\top \mathbf{w}(c_n,
\mathcal{S}_n))^2 \leq C.\]

By Corollary~\ref{cor.distributionalrobust} and $\hat{\mu}_n \in
\hat{\mathcal{P}}(n)$ we have
\begin{equation*}
\begin{split}
 &\sqrt{\int_{b,\mathbf{r}} (b-\mathbf{r}^\top \mathbf{x}(c_n,
\mathcal{S}_n))^2 d \hat{\mu}_n(b,\mathbf{r})}\\
\leq&\sup_{\mu\in \hat{\mathcal{P}}(n)} \sqrt{\int_{b,\mathbf{r}}
(b-\mathbf{r}^\top
\mathbf{x}(c_n, \mathcal{S}_n))^2 d \mu(b,\mathbf{r})} \\
=& \frac{\sqrt{n}}{n}\sqrt{\sum_{i=1}^n (b_i-\mathbf{r}_i^\top
\mathbf{x}(c_n, \mathcal{S}_n))^2}+c_n\|\mathbf{x}(c_n,
\mathcal{S}_n)\|_1+c_n\\
\leq & \frac{\sqrt{n}}{n}\sqrt{\sum_{i=1}^n (b_i-\mathbf{r}_i^\top
\mathbf{x}(\mathbb{P}))^2}+c_n\|\mathbf{x}(\mathbb{P})\|_1+c_n,
\end{split}
\end{equation*}
the last inequality holds by definition of $\mathbf{x}(c_n,
\mathcal{S}_n)$.

Taking the square of both sides, we have
\begin{equation*}\begin{split}&\int_{b,\mathbf{r}} (b-\mathbf{r}^\top \mathbf{x}(c_n,
\mathcal{S}_n))^2 d \hat{\mu}_n(b,\mathbf{r})\\\leq&
\frac{1}{n}\sum_{i=1}^n (b_i-\mathbf{r}_i^\top
\mathbf{x}(\mathbb{P}))^2+c_n^2
(1+\|\mathbf{x}(\mathbb{P})\|_1)^2\\&\quad\quad+2c_n(1+\|\mathbf{x}(\mathbb{P})\|_1)\sqrt{\frac{1}{n}\sum_{i=1}^n
(b_i-\mathbf{r}_i^\top \mathbf{x}(\mathbb{P}))^2}.
\end{split}\end{equation*}
Notice that, the right-hand side converges to $\int_{b,\mathbf{r}}
(b-\mathbf{r}^\top \mathbf{x}(\mathbb{P}))^2 d\mathbb{P}(b,
\mathbf{r})$ as $n\uparrow \infty$ and $c_n\downarrow 0$ almost
surely. Furthermore, we have
\begin{equation*}\begin{split}&\int_{b,\mathbf{r}} (b-\mathbf{r}^\top
\mathbf{x}(c_n, \mathcal{S}_n))^2 d\mathbb{P}(b,\mathbf{r})
\\\leq& \int_{b,\mathbf{r}}
(b-\mathbf{r}^\top \mathbf{x}(c_n, \mathcal{S}_n))^2 d
\hat{\mu}_n(b,\mathbf{r})\\&\quad\quad
+\big[\max_{b,\mathbf{r}}(b-\mathbf{r}^\top\mathbf{x}(c_n,
\mathcal{S}_n))^2\big] \int_{b,\mathbf{r}} |f_n(b,\mathbf{r})-f^*(b,
\mathbf{r})| d(b,\mathbf{r})\\\leq &\int_{b,\mathbf{r}}
(b-\mathbf{r}^\top \mathbf{x}(c_n, \mathcal{S}_n))^2 d
\hat{\mu}_n(b,\mathbf{r}) +C \int_{b,\mathbf{r}}
|f_n(b,\mathbf{r})-f^*(b, \mathbf{r})| d(b,\mathbf{r}),\end{split}
\end{equation*}
where the last inequality follows from the definition of $C$. Notice
that $\int_{b,\mathbf{r}} |f_n(b,\mathbf{r})-f^*(b, \mathbf{r})|
d(b,\mathbf{r})$ goes to zero almost surely when $c_n\downarrow 0$
and $nc_n^{m+1}\uparrow \infty$ since $f_n(\cdot)$ is a kernel
density estimation of $f^*(\cdot)$ (see e.g. Theorem 3.1 of
\cite{Devroye85}). Hence the theorem follows.
\end{proof}

We can remove the assumption that $\|\mathbf{x}(c_n,
\mathcal{S}_n)\|_2\leq H$, and as in Theorem~\ref{thm.consistencyforbounded},  the proof technique rather
than the result itself is of interest.
\begin{theorem}
\label{thm.consistencyunbdd}
Let $\{c_n\}$ converge to zero sufficiently slowly. Then
\[\lim_{n\rightarrow \infty} \sqrt{\int_{b,\mathbf{r}} (b-\mathbf{r}^\top \mathbf{x}(c_n, \mathcal{S}_n))^2 d\mathbb{P}(b,
\mathbf{r})}=\sqrt{\int_{b,\mathbf{r}} (b-\mathbf{r}^\top
\mathbf{x}(\mathbb{P}))^2 d\mathbb{P}(b, \mathbf{r})},\] almost
surely.
\end{theorem}
\begin{proof} To prove the theorem, we need to consider a set of distributions belonging to
$\hat{\mathcal{P}}(n)$. Hence we establish the following lemma
first.
\begin{lemma}\label{lem.partition}Partition the support of $\mathbb{P}$ as $V_1,\cdots, V_T$ such
the $\ell^\infty$ radius of each set is less than $c_n$. If a
distribution $\mu$ satisfies
\begin{equation}\label{equ.partition}\mu(V_t)=\Big|\big\{i|(b_i,\mathbf{r}_i)\in V_t\big\}\Big|/n;\quad
t=1,\cdots, T,\end{equation} then $\mu\in \hat{\mathcal{P}}(n)$.
\end{lemma}
\begin{proof}
Let $\mathcal{Z}_i=[b_i-c_n,b_i+c_n]\times\prod_{j=1}^m [a_{ij}-c_n,
a_{ij}+c_n]$; recall that $a_{ij}$ the $j^{th}$ element of
$\mathbf{r}_i$. Notice $V_t$ has $\ell^{\infty}$ norm less than $c_n$
we have
\[(b_i, \mathbf{r}_i\in V_t) \Rightarrow V_t\subseteq
\mathcal{Z}_i.\] Therefore, for any $S\subseteq \{1, \cdots, n\}$,
the following holds \begin{equation*}\begin{split}&\mu(\bigcup_{i\in
S} \mathcal{Z}_i)\geq \mu(\bigcup V_t|\exists i\in S: b_i,
\mathbf{r}_i\in V_t)\\=&\sum_{t|\exists i\in S: b_i, \mathbf{r}_i\in
V_t} \mu(V_t)=\sum_{t|\exists i\in S: b_i, \mathbf{r}_i\in V_t}\#
\big((b_i,\mathbf{r}_i)\in V_t\big)/n\geq |S|/n.
\end{split}\end{equation*}
Hence $\mu\in \mathcal{P}_n(A, \Delta, b, c_n)$ where each element
of $\Delta$ is $c_n$, which leads to $\mu\in \hat{\mathcal{P}}(n)$.
\end{proof}
Now we proceed to prove the theorem. Partition the support of
$\mathbb{P}$ into $T$ subsets such that $\ell^\infty$ radius of each
one is smaller than $c_n$. Denote $\tilde{\mathcal{P}}(n)$ as the
set of probability measures satisfying
Equation~(\ref{equ.partition}). Hence
$\tilde{\mathcal{P}}(n)\subseteq \hat{\mathcal{P}}(n)$ by
Lemma~\ref{lem.partition}. Further notice that there exists a
universal constant $K$ such that $\|\mathbf{x}(c_n,
\mathcal{S}_n)\|_2 \leq K/c_n$ due to the fact that the square loss
of the solution $\mathbf{x}=\mathbf{0}$ is bounded by a constant
only depends on the support of $\mathbb{P}$. Thus, there exists a
constant $C$ such that $\max_{b, \mathbf{r}}
(b-\mathbf{r}^\top\mathbf{x}(c_n, \mathcal{S}_n))^2\leq C/c_n^2$.

Follow a similar argument as the proof of
Theorem~\ref{thm.consistencyforbounded}, we have
\begin{equation}\begin{split}\label{equ.proofofconsisnobound}&\sup_{\mu_n\in \tilde{\mathcal{P}}(n)}\int_{b,\mathbf{r}} (b-\mathbf{r}^\top \mathbf{x}(c_n,
\mathcal{S}_n))^2 d \mu_n(b,\mathbf{r})\\\leq&
\frac{1}{n}\sum_{i=1}^n (b_i-\mathbf{r}_i^\top
\mathbf{x}(\mathbb{P}))^2+c_n^2
(1+\|\mathbf{x}(\mathbb{P})\|_1)^2\\&\quad\quad+2c_n(1+\|\mathbf{x}(\mathbb{P})\|_1)\sqrt{\frac{1}{n}\sum_{i=1}^n
(b_i-\mathbf{r}_i^\top \mathbf{x}(\mathbb{P}))^2},
\end{split}\end{equation}
and
\begin{equation*}\begin{split}&\int_{b,\mathbf{r}} (b-\mathbf{r}^\top
\mathbf{x}(c_n, \mathcal{S}_n))^2 d\mathbb{P}(b,\mathbf{r})
\\\leq&\inf_{\mu_n\in \tilde{\mathcal{P}}(n)}\Big\{ \int_{b,\mathbf{r}}
(b-\mathbf{r}^\top \mathbf{x}(c_n, \mathcal{S}_n))^2 d
\mu_n(b,\mathbf{r})\\&\quad\quad
+\max_{b,\mathbf{r}}(b-\mathbf{r}^\top\mathbf{x}(c_n,
\mathcal{S}_n))^2 \int_{b,\mathbf{r}} |f_{\mu_n}(b,\mathbf{r})-f(b,
\mathbf{r})| d(b,\mathbf{r})\\
\leq &\sup_{\mu_n\in \tilde{\mathcal{P}}(n)} \int_{b,\mathbf{r}}
(b-\mathbf{r}^\top \mathbf{x}(c_n, \mathcal{S}_n))^2 d
\mu_n(b,\mathbf{r})\\&\quad\quad +2C/c_n^2\inf_{\mu'_n\in
\tilde{\mathcal{P}}(n)}\Big\{ \int_{b,\mathbf{r}}
|f_{\mu'_n}(b,\mathbf{r})-f(b, \mathbf{r})|
d(b,\mathbf{r})\Big\},\end{split}
\end{equation*}
here $f_{\mu}$ stands for the density function of a measure $\mu$.
Notice that $\tilde{\mathcal{P}}_n$ is the set of distributions
satisfying Equation~(\ref{equ.partition}), hence $\inf_{\mu'_n\in
\tilde{\mathcal{P}}(n)}\int_{b,\mathbf{r}}
|f_{\mu'_n}(b,\mathbf{r})-f(b, \mathbf{r})| d(b,\mathbf{r})$ is
upper-bounded by $\sum_{t=1}^T |\mathbb{P}(V_t)-\# (b_i,
\mathbf{r}_i \in V_t)|/n$, which goes to zero as $n$ increases for
any fixed $c_n$ (see for example Proposition A6.6 of
\cite{Vaart2000}). Therefore,
\[2C/c_n^2\inf_{\mu'_n\in
\tilde{\mathcal{P}}(n)}\Big\{ \int_{b,\mathbf{r}}
|f_{\mu'_n}(b,\mathbf{r})-f(b, \mathbf{r})|
d(b,\mathbf{r})\Big\}\rightarrow 0,\] if $c_n\downarrow 0$
sufficiently slow. Combining this with
Inequality~(\ref{equ.proofofconsisnobound}) proves the theorem.
\end{proof}

\section{Stability}\label{sec.stability}
Knowing that the robust regression
problem~(\ref{equ.robustregression}) and in particular Lasso
encourage sparsity, it is of interest to investigate another
desirable characteristic of a learning algorithm, namely, stability.
We show in this section that Lasso {\em is not stable}. This is a
special case of a more general result we prove in
\cite{XuCaramanisMannorNFLallerton}, where we show that this is a
common property for all algorithms that encourage sparsity. That is,
if a learning algorithm achieves certain sparsity condition, then it
cannot have a non-trivial stability bound.

We recall the definition of uniform stability \cite{Bousquet02}
first. We let $\mathcal{Z}$ denote the space of points and labels
(typically this will be a compact subset of $\R^{n+1}$) so that $S
\in \mathcal{Z}^m$ denotes a collection of $m$ labelled training
points. We let $\mathbb{L}$ denote a learning algorithm, and for $S
\in \mathcal{Z}^m$, we let $\mathbb{L}_S$ denote the output of the
learning algorithm (i.e., the regression function it has learned
from the training data). Then given a loss function $l$, and a
labeled point $s =(\mathbf{z},b) \in \mathcal{Z}$, we let
$l(\mathbb{L}_S,s)$ denote the loss of the algorithm that has been
trained on the set $S$, on the data point $s$. Thus for squared
loss, we would have $l(\mathbb{L}_S,s) = \|\mathbb{L}_S(\mathbf{z})
- b \|_2$.

\begin{definition} An algorithm $\mathbb{L}$ has
uniform stability bound of $\beta_m$ with respect to the loss
function $l$ if the following holds
\[\forall S\in \mathcal{Z}^m, \forall i\in \{1,\cdots,m\}, \|l(\mathbb{L}_S,
\cdot)-l(\mathbb{L}_{S^{\backslash i}},\cdot)\|_\infty \leq
\beta_m.\]
Here $\mathbb{L}_{S^{\backslash i}}$ stands for the
learned solution with the $i^{th}$ sample removed from $S$.
\end{definition}
At first glance, this definition may seem too stringent for any
reasonable algorithm to exhibit good stability properties. However,
as shown in \cite{Bousquet02}, {\it Tikhonov-regularized regression
has stability that scales as $1/m$}. Stability that scales at least
as fast as $o(\frac{1}{\sqrt{m}})$ can be used to establish strong
PAC bounds (see \cite{Bousquet02}).

In this section we show that not only is the stability (in the sense defined above)
of Lasso much worse than the stability of $\ell^2$-regularized regression, but in
fact Lasso's stability is, in the following sense, as bad as it gets.
To this end, we define the notion of the trivial bound, which is the worst
possible error a training algorithm can have for arbitrary training
set and testing sample labelled by zero.
\begin{definition}\label{def.triaboundlasso}
Given a subset from which we can draw $m$ labelled points, $\mathcal{Z}\subseteq \mathbb{R}^{n\times (m+1)}$ and
a subset for one unlabelled point, $\mathcal{X}\subseteq \mathbb{R}^{m}$, a trivial bound
for a learning algorithm $\mathbb{L}$ w.r.t. $\mathcal{Z}$  and
$\mathcal{X}$ is
\[\mathfrak{b}(\mathbb{L}, \mathcal{Z}, \mathcal{X})\triangleq \max_{S\in \mathcal{Z}, \mathbf{z}\in \mathcal{X}} l\big(\mathbb{L}_{S},
(\mathbf{z},0)\big).\]
As above, $l(\cdot,\cdot)$ is a given loss function.
\end{definition}
Notice that the trivial bound does not diminish as the
number of samples increases, since by repeatedly choosing
the worst sample, the algorithm will yield the same solution.

Now we show that the uniform stability bound of Lasso can be no better than
its trivial bound with the number of features halved.
\begin{theorem}Given $\hat{\mathcal{Z}}\subseteq \mathbb{R}^{n\times
(2m+1)}$ be the domain of sample set and $\hat{\mathcal{X}}\subseteq
\mathbb{R}^{2m}$ be the domain of new observation, such that
\begin{equation*}\begin{split}&(\mathbf{b}, A)\in
\mathcal{Z}\Longrightarrow (\mathbf{b}, A, A)\in \hat{\mathcal{X}},\\
&(\mathbf{z}^\top)\in \mathcal{X}\Longrightarrow (\mathbf{z}^\top,
\mathbf{z}^\top)\in \hat{\mathcal{X}}.
\end{split}
\end{equation*}
Then  the uniform stability bound of Lasso is lower bounded by
$\mathfrak{b}(\mathrm{Lasso}, \mathcal{Z}, \mathcal{X})$.
\end{theorem}
\begin{proof} Let $(\mathbf{b}^*, A^*)$ and $(0, \mathbf{z}^{*\top})$
be the sample set and the new observation such that they jointly
achieve $\mathfrak{b}(\mathrm{Lasso}, \mathcal{Z}, \mathcal{X})$, and let
$\mathbf{x}^*$ be the optimal solution to Lasso w.r.t
$(\mathbf{b}^*, A^*)$. Consider the following sample set
\[\left(\begin{array}{ccc}\mathbf{b}^* & A^* & A^*\\ 0 &\mathbf{0}^\top &\mathbf{z}^{*\top}\end{array}
\right).\] Observe that $(\mathbf{x}^\top, \mathbf{0}^\top)^\top$ is
an optimal solution of Lasso w.r.t to this sample set. Now remove
the last sample from the sample set. Notice that $(\mathbf{0}^\top,
\mathbf{x}^\top)^\top$ is an optimal solution for this new sample
set. Using the last sample as a testing observation, the solution
w.r.t the full sample set has zero cost, while the solution of the
leave-one-out sample set has a cost $\mathfrak{b}(\mathrm{Lasso}, \mathcal{Z},
\mathcal{X})$. And hence we prove the theorem.
\end{proof}

\section{Conclusion}\label{sec.conclusion}
In this paper, we considered  robust regression with a
least-square-error loss. In contrast to previous work on robust
regression, we considered the case where the perturbations of the
observations are in the features. We show that this formulation is
equivalent to a weighted $\ell^1$ norm regularized regression
problem if no correlation of disturbances among different features
is allowed, and hence provide an interpretation of the widely used
Lasso algorithm from a robustness perspective. We also formulated
tractable robust regression problems for disturbance coupled among
different features and hence generalize Lasso to a wider class of
regularization schemes.

The sparsity and consistency of  Lasso are also investigated based
on its robustness interpretation. In particular we present a
``no-free-lunch'' theorem saying that sparsity and algorithmic
stability contradict each other. This result shows, although
sparsity and algorithmic stability are both regarded as desirable
properties of regression algorithms, it is not possible to achieve
them simultaneously, and we have to  tradeoff these two properties
in designing a regression algorithm.

The main thrust of this work is to treat the widely used regularized
regression scheme from a robust optimization perspective, and extend
the result of \cite{Ghaoui97} (i.e., Tikhonov regularization is
equivalent to a robust formulation for Frobenius norm bounded
disturbance set) to a broader range of disturbance set and hence
regularization scheme. This provides us not only with new insight of
why regularization schemes work,   but also offer solid motivations
for selecting regularization parameter for existing regularization
scheme and facilitate designing new regularizing schemes.

\bibliographystyle{unsrt}
\bibliography{Phd1}

\appendices

\section{Proof of Theorem~\ref{thm.probabound}}
{\bf Theorem~\ref{thm.probabound}}. {\it Consider a random vector $\mathbf{v}\in \mathbb{R}^n$,
such that $\mathbb{E}(\mathbf{v})=\mathbf{a}$, and
$\mathbb{E}(\mathbf{v} \mathbf{v}^\top)=\Sigma$, $\Sigma \succeq 0$.
Then we have
\begin{equation}\label{equ.probbound}\Pr\{\|\mathbf{v}\|_2\geq c_i\} \leq
\left\{\begin{array}{rl} \min_{P, \mathbf{q}, r,
\lambda}\,\,&\mathrm{Trace}(\Sigma P)+2\mathbf{q}^\top\mathbf{a}+
r\\\mbox{subject to:}\,\,&\left(\begin{array}{ll}P & \mathbf{q} \\ \mathbf{q}^\top & r\end{array}\right)\succeq 0\\
 &\left(\begin{array}{ll} I(m) &
\mathbf{0} \\ \mathbf{0}^\top &-c^2_i\end{array}\right)\preceq
\lambda \left(\begin{array}{ll} P & \mathbf{q}\\ \mathbf{q}^\top &
r-1\end{array} \right)\\ &\lambda \geq
0.\end{array}\right.\end{equation}}
\begin{proof}
Consider a function $f(\cdot)$ parameterized by $P, \mathbf{q}, r$
defined as $f(\mathbf{v})=\mathbf{v}^\top P\mathbf{v} +2
\mathbf{q}^\top \mathbf{v}+r$. Notice
$\mathbb{E}\big(f(\mathbf{v})\big)=\mathrm{Trace}(\Sigma
P)+2\mathbf{q}^\top\mathbf{a}+ r$. Now we show that
$f(\mathbf{v})\geq \mathbf{1}_{\|\mathbf{v}\|\geq c_i}$ for all $P,
\mathbf{q}, r$ satisfying the constraints in (\ref{equ.probbound}).

To show $f(\mathbf{v})\geq \mathbf{1}_{\|\mathbf{v}\|_2\geq c_i}$,
we need to establish (i) $f(\mathbf{v})\geq 0$ for all $\mathbf{v}$,
and (ii) $f(\mathbf{v})\geq 1 $ when $\|\mathbf{v}\|_2\geq c_i$.
Notice that \[f(\mathbf{v}) =\left(\begin{array}{c}\mathbf{v}
\\ 1\end{array}\right)^\top\left(\begin{array}{ll}P & \mathbf{q}
\\ \mathbf{q}^\top & r\end{array}\right)\left(\begin{array}{c}\mathbf{v}
\\ 1\end{array}\right),\] hence (i) holds because \[\left(\begin{array}{ll}P & \mathbf{q} \\ \mathbf{q}^\top & r\end{array}\right)\succeq
0.\]

To establish condition (ii), it suffices to show
$\mathbf{v}^\top\mathbf{v}\geq c^2_i$ implies $\mathbf{v}^\top P
\mathbf{v} +2 \mathbf{q}^\top \mathbf{v} +r\geq 1$, which is
equivalent to show $\big\{\mathbf{v}\big|\mathbf{v}^\top P
\mathbf{v} +2 \mathbf{q}^\top \mathbf{v} +r- 1\leq 0\big\}\subseteq
\big\{\mathbf{v}\big|\mathbf{v}^\top \mathbf{v}\leq c_i^2\big\}$.
Noticing this is an ellipsoid-containment condition, by S-Procedure,
we see that is equivalent to the condition that there exists a
$\lambda \geq 0$ such that
\[\left(\begin{array}{ll} I(m) &
\mathbf{0} \\ \mathbf{0}^\top &-c^2_i\end{array}\right)\preceq
\lambda \left(\begin{array}{ll} P & \mathbf{q}\\ \mathbf{q}^\top &
r-1\end{array} \right).\]

Hence we have  $f(\mathbf{v})\geq \mathbf{1}_{\|\mathbf{v}\|_2\geq
c_i}$, taking expectation over both side that notice that the
expectation of a indicator function is the probability, we establish
the theorem.
\end{proof}

\section{Proof of Theorem~\ref{thm.generaluncertainty}}
{\bf Theorem~\ref{thm.generaluncertainty}}. {\it
Assume that the set \[\mathcal{Z}\triangleq\{\mathbf{z}\in
\mathbb{R}^m| f_j(\mathbf{z})\leq 0,\,\,j=1,\cdots, k;\,\,
\mathbf{z}\geq \mathbf{0}\}\] has non-empty relative interior. Then the
robust regression problem
\[\min_{\mathbf{x}\in \mathbb{R}^m} \left\{\max_{\Delta A \in
\mathcal{U}'}\|\mathbf{b}-(A+\Delta A)\mathbf{x}\|_a\right\}\] is
equivalent to the following regularized regression problem
\begin{equation*}\label{equ.generaluncertaintyinappendix}\begin{split}&\min_{\boldsymbol{\lambda}\in
\mathbb{R}^k_+,\boldsymbol{\kappa}\in \mathbb{R}^m_+, \mathbf{x}\in
\mathbb{R}^m}\Big\{
\|\mathbf{b}-A\mathbf{x}\|_a+v(\boldsymbol{\lambda},\boldsymbol{\kappa},
\mathbf{x})\Big\};\\&\mbox{where:}\,\,
v(\boldsymbol{\lambda},\boldsymbol{\kappa},\mathbf{x})\triangleq
\max_{\mathbf{c}\in
\mathbf{R}^m}\Big[(\boldsymbol{\kappa}+|\mathbf{x}|)^\top
\mathbf{c}-\sum_{j=1}^k\lambda_j
f_j(\mathbf{c})\Big]\end{split}\end{equation*} }
\begin{proof}Fix a solution $\mathbf{x}^*$. Notice that,
\[\mathcal{U}'=\{(\boldsymbol{\delta}_1,\cdots, \boldsymbol{\delta}_m)|\mathbf{c}\in \mathcal{Z};\,\, \|\boldsymbol{\delta}_i\|_a\leq c_i,\,i=1,\cdots,
m\}.\] Hence we have:
\begin{equation}\label{equ.proofinmostgeneral}
\begin{split}
&\max_{\Delta A \in \mathcal{U}'}\|\mathbf{b}-(A+\Delta
A)\mathbf{x}^*\|_a\\
=&\max_{\mathbf{c}\in
\mathcal{Z}}\Big\{\max_{\|\boldsymbol{\delta}_i\|_a\leq
c_i,\,i=1,\cdots,m}
\|\mathbf{b}-\big(A+(\boldsymbol{\delta}_1,\cdots,\boldsymbol{\delta}_m)\big)\mathbf{x}^*\|_a\Big\}\\
=&\max_{\mathbf{c}\in
\mathcal{Z}}\Big\{\|\mathbf{b}-A\mathbf{x}^*\|_a+\sum_{i=1}^m
c_i|x_i^*|\Big\}\\=&\|\mathbf{b}-A\mathbf{x}^*\|_a+\max_{\mathbf{c}\in
\mathcal{Z}}\Big\{|\mathbf{x}^*|^\top \mathbf{c}\Big\}.
\end{split}
\end{equation}
The second equation follows from Theorem~\ref{thm.generalnorm}.

Now we need to evaluate $\max_{\mathbf{c}\in
\mathcal{Z}}\{|\mathbf{x}^*|^\top \mathbf{c}\}$, which equals to
$-\min_{\mathbf{c}\in \mathcal{Z}}\{-|\mathbf{x}^*|^\top
\mathbf{c}\}$. Hence we are minimizing a linear function over a set
of convex constraints. Furthermore, by assumption the Slater's
condition holds. Hence the duality gap of $\min_{\mathbf{c}\in
\mathcal{Z}}\{-|\mathbf{x}^*|^\top \mathbf{c}\}$ is zero. A standard
duality analysis shows that
\begin{equation}\label{equ.proofinmg2}\max_{\mathbf{c}\in
\mathcal{Z}}\Big\{|\mathbf{x}^*|^\top \mathbf{c}\Big\}
=\min_{\boldsymbol{\lambda}\in \mathbb{R}^k_+,\boldsymbol{\kappa}\in
\mathbb{R}^m_+}v(\boldsymbol{\lambda},\boldsymbol{\kappa},\mathbf{x}^*).
\end{equation} We establish the theorem by substituting Equation~(\ref{equ.proofinmg2}) back into
Equation~(\ref{equ.proofinmostgeneral}) and taking minimum over
$\mathbf{x}$ on both sides.
\end{proof}

\section{Proof of Proposition~\ref{prop.pointvsdist}}
{\bf Proposition~\ref{prop.pointvsdist}}. {\it Given a function $g:\mathbb{R}^{m+1}\rightarrow \mathbb{R}$
and Borel sets $\mathcal{Z}_1,\cdots, \mathcal{Z}_n\subseteq
\mathbb{R}^{m+1}$, let
\[\mathcal{P}_n\triangleq \{\mu\in \mathcal{P}| \forall S\subseteq \{1,\cdots, n\}: \mu(\bigcup_{i\in S} \mathcal{Z}_i)\geq |S|/n \}.\]
The following holds
\[\frac{1}{n}\sum_{i=1}^n \sup_{(\mathbf{r}_i,b_i)\in \mathcal{Z}_i} h(\mathbf{r}_i,b_i) =\sup_{\mu\in \mathcal{P}_n}\int_{\mathbb{R}^{m+1}} h(\mathbf{r},b)d\mu(\mathbf{r},b). \]
}
\begin{proof}
To prove Proposition~\ref{prop.pointvsdist}, we first establish the following
lemma.
\begin{lemma}Given a function $f: \mathbb{R}^{m+1} \rightarrow
\mathbb{R}$, and a Borel set $\mathcal{Z}\subseteq
\mathbb{R}^{m+1}$, the following holds:
\[\sup_{\mathbf{x}'\in \mathcal{Z}} f(\mathbf{x}')=\sup_{\mu\in\mathcal{P}|
\mu(\mathcal{Z})=1}\int_{\mathbb{R}^{m+1}}
f(\mathbf{x})d\mu(\mathbf{x}).\]
\end{lemma}
\begin{proof}Let $\hat{\mathbf{x}}$ be a $\epsilon-$optimal solution
to the left hand side, consider the probability measure $\mu'$ that
put mass $1$ on $\hat{\mathbf{x}}$, which satisfy
$\mu'(\mathcal{Z})=1$. Hence, we have
\[\sup_{\mathbf{x}'\in \mathcal{Z}} f(\mathbf{x}')-\epsilon \leq\sup_{\mu\in\mathcal{P}|
\mu(\mathcal{Z})=1}\int_{\mathbb{R}^{m+1}}
f(\mathbf{x})d\mu(\mathbf{x}),\] since $\epsilon$ can be arbitrarily
small, this leads to
\begin{equation}\label{equ.pointlessprob}\sup_{\mathbf{x}'\in \mathcal{Z}}
f(\mathbf{x}') \leq\sup_{\mu\in\mathcal{P}|
\mu(\mathcal{Z})=1}\int_{\mathbb{R}^{m+1}}
f(\mathbf{x})d\mu(\mathbf{x}).\end{equation} Next construct function
$\hat{f}:\mathbb{R}^{m+1}\rightarrow \mathbb{R}$ as
\[\hat{f}(\mathbf{x})\triangleq \left\{\begin{array}{ll}
f(\hat{\mathbf{x}}) & \mathbf{x}\in \mathcal{Z}; \\ f(\mathbf{x}) &
\mbox{otherwise}. \end{array}\right.\] By definition of
$\hat{\mathbf{x}}$ we have $f(\mathbf{x})\leq
\hat{f}(\mathbf{x})+\epsilon$ for all $\mathbf{x}\in
\mathbb{R}^{m+1}$. Hence, for any probability measure $\mu$ such
that $\mu(\mathcal{Z})=1$, the following holds
\[\int_{\mathbb{R}^{m+1}} f(\mathbf{x})d\mu(x)\leq  \int_{\mathbb{R}^{m+1}}
\hat{f}(\mathbf{x})d\mu(x)+\epsilon=f(\hat{\mathbf{x}})+\epsilon\leq
\sup_{\mathbf{x}'\in \mathcal{Z}} f(\mathbf{x}')+\epsilon.\] This
leads to
\[\sup_{\mu\in \mathcal{P}|\mu(\mathcal{Z})=1}\int_{\mathbb{R}^{m+1}} f(\mathbf{x})d\mu(x)\leq \sup_{\mathbf{x}'\in \mathcal{Z}}
f(\mathbf{x}')+\epsilon.\] Notice $\epsilon$ can be arbitrarily
small, we have \begin{equation}\label{equ.problesspoint}\sup_{\mu\in
\mathcal{P}|\mu(\mathcal{Z})=1}\int_{\mathbb{R}^{m+1}}
f(\mathbf{x})d\mu(x)\leq \sup_{\mathbf{x}'\in \mathcal{Z}}
f(\mathbf{x}')\end{equation}Combining (\ref{equ.pointlessprob}) and
(\ref{equ.problesspoint}), we prove the lemma.
\end{proof}
Now we proceed to prove the proposition. Let $\hat{\mathbf{x}}_i$ be
an $\epsilon-$optimal solution to $\sup_{\mathbf{x}_i\in
\mathcal{Z}_i} f(\mathbf{x}_i)$. Observe that the empirical
distribution for $(\hat{\mathbf{x}}_1,\cdots, \hat{\mathbf{x}}_n)$
belongs to $\mathcal{P}_n$, since $\epsilon$ can be arbitrarily
close to zero, we have
\begin{equation}\label{equ.them.pointsmallthandis}\frac{1}{n}\sum_{i=1}^n
\sup_{\mathbf{x}_i\in \mathcal{Z}_i} f(\mathbf{x}_i)
\leq\sup_{\mu\in \mathcal{P}_n}\int_{\mathbb{R}^{m+1}}
f(\mathbf{x})d\mu(\mathbf{x}). \end{equation}

Without loss of generality, assume
\begin{equation}\label{equ.increasewithi}f(\hat{\mathbf{x}}_1)\leq
f(\hat{\mathbf{x}}_2)\leq\cdots\leq
f(\hat{\mathbf{x}}_n).\end{equation} Now construct the following
function
\begin{equation}\hat{f}(\mathbf{x})\triangleq \left\{\begin{array}{ll}
\min_{i|\mathbf{x}\in \mathcal{Z}_i}f(\hat{\mathbf{x}}_i) &
\mathbf{x}\in \bigcup_{j=1}^n \mathcal{Z}_j;\\ f(\mathbf{x}) &
\mbox{otherwise}.
\end{array}\right.\end{equation}
Observe that $f(\mathbf{x})\leq \hat{f}(\mathbf{x})+\epsilon$ for
all $\mathbf{x}$.

Furthermore, given $\mu\in \mathcal{P}_n$, we have
\begin{equation*}
\begin{split}&\int_{\mathbb{R}^{m+1}} f(\mathbf{x})d\mu(\mathbf{x})-\epsilon\\
=&\int_{\mathbb{R}^{m+1}}
\hat{f}(\mathbf{x})d\mu(\mathbf{x})\\
=& \sum_{k=1}^n f(\hat{\mathbf{x}}_k)\Big[\mu(\bigcup_{i=1}^k
\mathcal{Z}_i)-\mu(\bigcup_{i=1}^{k-1} \mathcal{Z}_i)\Big]
\end{split}
\end{equation*}
Denote $\alpha_k\triangleq \Big[\mu(\bigcup_{i=1}^k
\mathcal{Z}_i)-\mu(\bigcup_{i=1}^{k-1} \mathcal{Z}_i)\Big]$, we have
\[\sum_{k=1}^n {\alpha_k}=1,\quad \sum_{k=1}^t \alpha_k\geq
t/n.\]Hence by Equation~(\ref{equ.increasewithi}) we have
\[\sum_{k=1}^n {\alpha_k}f(\hat{\mathbf{x}}_k)\leq \frac{1}{n}
\sum_{k=1}^n f(\hat{\mathbf{x}}_k).\] Thus we have for any $\mu\in
\mathcal{P}_n$,
\[\int_{\mathbb{R}^{m+1}} f(\mathbf{x})d\mu(\mathbf{x})-\epsilon\leq \frac{1}{n}
\sum_{k=1}^n f(\hat{\mathbf{x}}_k).\] Therefore,
\[\sup_{\mu\in \mathcal{P}_n}\int_{\mathbb{R}^{m+1}} f(\mathbf{x})d\mu(\mathbf{x})-\epsilon\leq \sup_{\mathbf{x}_i\in \mathcal{Z}_i}\frac{1}{n}
\sum_{k=1}^n f(\mathbf{x}_k).\] Notice $\epsilon$ can be arbitrarily
close to $0$, we proved the proposition by combining
with~(\ref{equ.them.pointsmallthandis}).\end{proof}

\section{Proof of Corollary~\ref{cor.distributionalrobust}}
{\bf Corollary~\ref{cor.distributionalrobust}.} {\it Given
$\mathbf{b}\in \mathbb{R}^n$, $A\in \mathbb{R}^{n\times m}$, the
following equation holds for any $\mathbf{x}\in \mathbb{R}^m$,
\begin{equation}\label{equ.distributionrobustinappendix}
\|\mathbf{b}-A\mathbf{x}\|_2  +\sqrt{n}c_n\|\mathbf{x}\|_1+
\sqrt{n}c_n =\sup_{\mu\in\hat{\mathcal{P}}(n)}
\sqrt{n\int_{\mathbb{R}^{m+1}}
(b'-\mathbf{r}'^\top\mathbf{x})^2d\mu(\mathbf{r}',b')}.\end{equation}Here,
\begin{equation*}
\begin{split}
\hat{\mathcal{P}}(n)&\triangleq
\bigcup_{\|\boldsymbol{\sigma}\|_2\leq\sqrt{n}c_n;\,\forall i:
\|\boldsymbol{\delta}_i\|_2 \leq \sqrt{n}c_n}\mathcal{P}_n(A,
\Delta, \mathbf{b},
\boldsymbol{\sigma});\\
\mathcal{P}_n (A, \Delta, \mathbf{b},
\boldsymbol{\sigma})&\triangleq \{\mu\in \mathcal{P} |
\mathcal{Z}_i=[b_i-\sigma_i,b_i+\sigma_i]\times\prod_{j=1}^m
[a_{ij}-\delta_{ij}, a_{ij}+\delta_{ij}];\\ &\quad\quad\forall
S\subseteq \{1,\cdots, n\}: \mu(\bigcup_{i\in S} \mathcal{Z}_i)\geq
|S|/n \}.\end{split}\end{equation*}}
\begin{proof}
The right-hand-side of
Equation~(\ref{equ.distributionrobustinappendix}) equals
\[\sup_{\|\boldsymbol{\sigma}\|_2\leq\sqrt{n}c_n;\,\forall i:
\|\boldsymbol{\delta}_i\|_2 \leq
\sqrt{n}c_n}\Big\{\sup_{\mu\in\mathcal{P}_n(A, \Delta, \mathbf{b},
\boldsymbol{\sigma})} \sqrt{n\int_{\mathbb{R}^{m+1}}
(b'-\mathbf{r}'^\top\mathbf{x})^2d\mu(\mathbf{r}',b')}\Big\}.\]Notice
by the equivalence to robust formulation, the left-hand-side equals
to
\begin{equation*}
\begin{split}
&\max_{\|\boldsymbol\sigma\|_2\leq \sqrt{n}c_n; \forall i:
\|\boldsymbol{\delta}_i\|_2\leq \sqrt{n}c_n}
\Big\|\mathbf{b}+\boldsymbol{\sigma}
-\big(A+[\boldsymbol{\delta}_1,\cdots,\boldsymbol{\delta}_m]
\big)\mathbf{x}\Big\|_2
\\=&\sup_{\|\boldsymbol{\sigma}\|_2\leq\sqrt{n}c_n;\,\forall i:
\|\boldsymbol{\delta}_i\|_2 \leq \sqrt{n}c_n}
\left\{\sup_{(\hat{b}_i, \hat{\mathbf{r}}_i)\in
[b_i-\sigma_i,b_i+\sigma_i]\times\prod_{j=1}^m [a_{ij}-\delta_{ij},
a_{ij}+\delta_{ij}]}
\sqrt{\sum_{i=1}^n(\hat{b}_i-\hat{\mathbf{r}}_i^\top\mathbf{x})^2}\right\}
\\=&\sup_{\|\boldsymbol{\sigma}\|_2\leq\sqrt{n}c_n;\,\forall i:
\|\boldsymbol{\delta}_i\|_2 \leq \sqrt{n}c_n}\sqrt{\sum_{i=1}^n
\sup_{(\hat{b}_i, \hat{\mathbf{r}}_i)\in
[b_i-\sigma_i,b_i+\sigma_i]\times\prod_{j=1}^m [a_{ij}-\delta_{ij},
a_{ij}+\delta_{ij}]}
(\hat{b}_i-\hat{\mathbf{r}}_i^\top\mathbf{x})^2},\end{split}\end{equation*}
furthermore, applying Proposition~\ref{prop.pointvsdist} yields
\begin{equation*}
\begin{split}&\sqrt{\sum_{i=1}^n \sup_{(\hat{b}_i, \hat{\mathbf{r}}_i)\in
[b_i-\sigma_i,b_i+\sigma_i]\times\prod_{j=1}^m [a_{ij}-\delta_{ij},
a_{ij}+\delta_{ij}]}
(\hat{b}_i-\hat{\mathbf{r}}_i^\top\mathbf{x})^2}\\=&\sqrt{\sup_{\mu\in\mathcal{P}_n(A,
\Delta , \mathbf{b}, \boldsymbol{\sigma})} n\int_{\mathbb{R}^{m+1}}
(b'-\mathbf{r}'^\top\mathbf{x})^2d\mu(\mathbf{r}',b')}\\
=&\sup_{\mu\in\mathcal{P}_n(A, \Delta , \mathbf{b},
\boldsymbol{\sigma})}\sqrt{ n\int_{\mathbb{R}^{m+1}}
(b'-\mathbf{r}'^\top\mathbf{x})^2d\mu(\mathbf{r}',b')},\end{split}\end{equation*}
which proves the corollary.
\end{proof}
\end{document}